\documentclass[12pt,final,
thmsa,%titlepage,
epsf,a4paper]%{article}
{article}
\usepackage{booktabs}

\newtheorem{theorem}{Theorem}%[chapter]

\newtheorem{corollary}{Corollary}%[chapter]
\newtheorem{definition}{Definition}%[chapter]
\newtheorem{example}{Example}%[chapter]
\newtheorem{lemma}{Lemma}%[chapter]
\newtheorem{proposition}{Proposition}%[chapter]
\newtheorem{remark}{Remark}%[chapter]
\newenvironment{proof}[1][Proof]{\emph{#1.} }{\  \hfill $\square $ \vspace{5 pt}}

\usepackage{stackrel}
\usepackage{natbib}
\usepackage{amssymb}
\usepackage[dvips]{graphics}
\usepackage{amsmath}

\usepackage{mathpazo}

\usepackage{enumerate}

\usepackage{amsfonts}

\usepackage{amssymb}

\usepackage{enumerate}

\usepackage{mathrsfs}

\usepackage[usenames]{color}

\usepackage{pgf,tikz}

\usetikzlibrary{decorations.pathreplacing,angles,quotes}
\usetikzlibrary{arrows.meta}
\usetikzlibrary{arrows, decorations.markings}

\usepackage{hyperref}
\hypersetup{
    colorlinks,
    citecolor=blue,
    filecolor=black,
    linkcolor=blue,
    urlcolor=blue
}
\usepackage{soul}
\usepackage{pgf,tikz}
\usetikzlibrary{arrows}

\usepackage[utf8]{inputenc}
\usepackage{amssymb, amsmath,geometry}
\usepackage{fancyhdr}
\usepackage{soul}
\usepackage{cancel}

\usepackage{emptypage}

\newcommand*\samethanks[1][\value{footnote}]{\footnotemark[#1]}

\begin{document}

\title{The lattice of worker-quasi-stable matchings\thanks{We thank Jordi Massó, Alejandro Neme, Changyong Hu, and  anonymous referees  
for their very detailed comments. We acknowledge financial support
from UNSL through grants 032016 and 030320, and from Consejo Nacional
de Investigaciones Cient\'{\i}ficas y T\'{e}cnicas (CONICET) through grant
PIP 112-200801-00655, and from Agencia Nacional de Promoción Cient\'ifica y Tecnológica through grant PICT 2017-2355.}}


\author{Agustín G. Bonifacio\thanks{Instituto de Matem\'{a}tica Aplicada San Luis, Universidad Nacional de San
Luis and CONICET, San Luis, Argentina, and RedNIE. Emails: \texttt{abonifacio@unsl.edu.ar} (A. G. Bonifacio), \texttt{ncguinazu@unsl.edu.ar} (N. Gui\~{n}azu), \texttt{nmjuarez@unsl.edu.ar} (N. Juarez), \texttt{paneme@unsl.edu.ar} (P. Neme) and \texttt{joviedo@unsl.edu.ar} (J. Oviedo).} \and  Nadia Gui\~{n}az\'u \samethanks[2] \and Noelia Juarez\samethanks[2] \and Pablo Neme\samethanks[2] \and Jorge Oviedo\samethanks[2]}

\date{\today}
\maketitle

\begin{abstract}

In a many-to-one matching model, we study the set of worker-quasi-stable matchings when firms' choice functions satisfy substitutability. Worker-quasi-stability is a relaxation of stability that allows blocking pairs involving a firm and an unemployed worker. We show that this set has a lattice structure and define a Tarski operator on this lattice that models a re-equilibration process and  has the set of stable matchings as its fixed points.  

\bigskip

\bigskip


\end{abstract}

\section{Introduction}
\label{intro}

In this paper, we study a many-to-one matching model in which agents in one side of the market (that we call \emph{firms}) have to be assigned to subsets of agents on the other side of the market (that we call \emph{workers}) and the only requirement on subsets of workers that each firm's choice function has to satisfy is substitutability. For this model, using a partial order first studied by \cite{blair1988lattice},  we show that the set of worker-quasi-stable matchings  has a lattice structure.\footnote{The lattice structure of the set of stable matchings is introduced by \cite{knuth1976marriages} for the one-to-one matching model. This result is generalized in different directions by several papers \citep[see for instance][among others]{blair1988lattice,marti2001lattice,alkan2002class,wu2018lattice}.} Worker-quasi-stability is a relaxation of stability that allows blocking pairs involving a firm and an unemployed worker.\footnote{In the many-to-one matching literature, stability can be thought of as the conjunction of ``envy-freeness'' and ``non-wastefulness'' \citep[see, for example,][]{Kamada2019fair,wu2018lattice}. The presence of blocking pairs involving an unemployed worker and an acceptant firm can be interpreted as wastefulness of the matching. } The importance of these matchings is twofold. First, from a practical standpoint,  worker-quasi-stability  captures an interim situation when, starting from a stable matching, new workers arrive or firms downsize and laid off workers have to look for employment elsewhere. Furthermore, if our aim is to pursue stability, a re-equilibration process  can be described as a layoff  chain dynamic within worker-quasi-stable matchings that brings the market back to stability. Second, from a theoretical standpoint, much of the literature that studies stability through Tarski's fixed point theorem \citep{tarski1955lattice}  carries out its analysis by means of a lattice that strictly contains the  set of matchings \citep[see][among others]{adachi2000characterization, fleiner2003fixed, echenique2004core}.\footnote{Most of these  papers rely on the notion of ``pre-matching''. In a pre-matching,   the fact that agent $a$ is matched to agent $b$ does not imply that $b$ is matched to $a.$} However, we show that a fixed point approach can be performed within the realm of worker-quasi-stability.

Our general framework assumes substitutability on firms' choice functions. This condition, first introduced by \cite{kelso1982job}, is the less restrictive requirement in firms' choice functions in order to guarantee the existence of stable matchings. A firm has substitutable choice functions if it wants to continue hiring a worker even if other workers become unavailable.  
\cite{blair1988lattice} defines a partial order over the set of matchings, and shows that when choice functions are substitutable, the set of stable matchings has a lattice structure. A matching Blair-dominates another matching   if each firm wishes to keep the workers hired under the first one even if all the workers hired under the second one are also available, and does not wish to hire any new worker.
 In our paper, given two worker-quasi-stable matchings, we define a new one by means of a choice function   that selects, for each firm, the best subset of workers among those that this firm is matched to in either matching. This new matching turns out to be the join (least upper bound) of the two original matchings according to Blair's order. In this way, we extend Blair's result to the whole set of worker-quasi-stable matchings. More specifically, we prove that the set of worker-quasi-stable matchings forms a finite join-semilattice with a minimum element, implying that it is a lattice.

Furthermore, we define a Tarski operator in the worker-quasi-stable lattice that describes a possible re-equilibration process. This process models how,  
starting from \textit{any} worker-quasi-stable matching,  a decentralized sequence of offers in which unemployed workers are hired (causing  new unemployments), produces a sequence of worker-quasi-stable matchings that converges to a stable matching. As a by-product, applying Tarski's fixed point theorem to our operator, we give an alternative proof of the fact that the set of stable matchings (the fixed point set of our operator) is non-empty and has a lattice structure as well. 
Finally, we present some additional results when   firms' choice functions satisfy, in addition to substitutability, the ``law of aggregate demand".\footnote{This property is first studied by \cite{alkan2002class} under the name of ``cardinal monotonicity". See also \cite{hatfield2005matching}.}  This condition says that when a firm chooses from an expanded set, it hires at least as many workers as before. 
Under the ``law of aggregate demand" we can  identify the fixed point of the operator that can be obtained by iterating it starting at a worker-quasi-stable matching: it is the join of that worker-quasi-stable matching and the worker-optimal stable matching. We also show that (i) the join of a worker-quasi-stable matching  and a stable matching is stable, and (ii) every worker-quasi-stable matching that weakly Blair-dominates the worker-optimal
stable matching is stable.

The paper closest to ours is \cite{wu2018lattice}. In a many-to-one matching model  in which firms have responsive preferences (a more restrictive requirement than substitutable choice functions), they obtain a lattice structure  for the set of firm-quasi-stable matchings\footnote{Firm-quasi-stability (called \textit{envy-freeness} by \cite{wu2018lattice}) is a relaxation of stability that allows blocking pairs involving a worker and an empty position of a firm.} under the common preference of workers. Given two matchings, the Conway-like join for workers that they use matches all workers to their most preferred firm between their two original  partners. Their paper is the first one to present a Tarski operator defined on a lattice of \emph{matchings} (in their case, the set of firm-quasi-stable matchings), that can be interpreted as modeling vacancy chains, and  show that the operator has the set of stable matchings as its fixed points.  Concerning the set of worker-quasi-stable matchings, they show the difficulties of defining a Conway-like join for firms even when firms have responsive preferences. Instead, we are able to sidestep this problem following Blair's insight.    

Another paper that relates a Tarski operator with the notion of firm-quasi-stability is \cite{Kamada2019fair}. In a school choice setting with constraints, they  define a Tarski operator over a space of ``cutoff profiles'', and use it to characterize the  firm-quasi-stable matchings\footnote{For \cite{Kamada2019fair}, a firm-quasi-stable matching is an \emph{envy-free} matching that fulfills a pre-specified constraint.} as its fixed points. Thus, their approach is tangential to ours and Wu and Roth's. 

The rest of the paper is organized as follows. In Section \ref{preliminar}, we present the model and preliminaries. The lattice structure of the worker-quasi-stable set is analyzed in Section \ref{lattice}. In Section \ref{seccion operador }, we introduce our Tarski operator, which allows us to prove that the set of stable matchings is non-empty and forms a lattice when firms' choice functions satisfy substitutability. Moreover, we present a re-equilibration process via layoff chains based on our operator. 
Further results that give some insight on the behavior of the Tarski operator are gathered in Section \ref{mas resultados LAD}, where in addition to substitutability we require firms' choice functions to satisfy the ``law of aggregate demand''. Finally, in Section \ref{concludings}, we present some conclusions.

\section{Model and preliminaries}\label{preliminar}

We consider a many-to-one matching model where there are two disjoint sets of agents: the set of \textit{firms} $F$ and the set of \textit{workers} $W$.  Each worker $w\in W$ has a strict preference relation $\succ_w$ over the individual firms and the prospect of being unmatched,  denoted by $\emptyset$. Each firm $f\in F$ has a choice function $C_f$ over the set of all subsets of $W$ that satisfies \textbf{substitutability}, i.e.,   for $S'\subseteq S \subseteq W$, we have $C_f(S)\cap S' \subseteq C_f(S')$.\footnote{Substitutability is equivalent to the following: for each $w\in W$ and each $S\subseteq W$ such that $w\in S$, $w\in C_f\left(S\right)$ implies that $w\in C_f\left(S'\cup \{w\}\right)$ for each $S' \subseteq S.$} In addition, we assume that $C_f$ satisfies $C_f(S')=C_f(S)$ whenever $C_f(S)\subseteq S' \subseteq S \subseteq W.$ This property is known in the literature as \textbf{consistency}. If $C_f$  satisfies substitutability and consistency, then it also satisfies \begin{equation}\label{propiedad de choice}
C_f\left(S\cup S'\right)=C_f\left(C_f\left(S\right)\cup S'\right)
\end{equation} for each pair of subsets $S$ and $S'$ of $W$.\footnote{This property is known in the literature as \emph{path independence} \citep[see][]{alkan2002class}.}

Let $\succ_W$ be the preference profile for all workers, and let $C_F$ be the profile of choice functions for all firms. A \emph{many-to-one matching market} is denoted by $(W,F,\succ_W,C_F).$ 

\begin{definition}
A \textbf{matching} $\mu$ is a function from  set $F\cup W$ into $2^{F\cup W}$ such that, for each $w\in W$ and  each $f\in F$:
\begin{enumerate}[(i)]
\item $\mu(w)\subseteq F$ with $|\mu(w)|\leq 1.$ 
\item $\mu(f)\subseteq W$.
\item $w\in \mu(f)$ if and only if $\mu(w)=\{f\}$. 
\end{enumerate}
\end{definition}
Usually, we will omit the curly brackets. For instance, instead of condition (iii) we will write: ``$w\in \mu(f)$ if and only if $\mu(w)=f$''.

Agent $a\in F\cup W$ is \textbf{matched} if $\mu(a) \neq \emptyset$, otherwise  $a$ is \textbf{unmatched}.
A matching $\mu$ is \textbf{blocked by a worker $\boldsymbol{w}$} if $\emptyset \succ_w \mu(w)$; that is, worker $w$ prefers being unemployed  rather than working  for firm $\mu(w)$. Similarly, $\mu$ is \textbf{blocked by a firm $\boldsymbol{f}$} if $\mu(f)\neq C_f\left(\mu(f)\right)$; that is, firm $f$ wants to fire some workers in $\mu(f)$. A matching is \textbf{individually rational} if it is not blocked by any individual agent.

A matching $\mu$ is \textbf{blocked by a firm-worker pair $\boldsymbol{(f,w)}$} if   $w \in C_f\left(\mu( f )\cup \{w\}\right),$  and $f \succ_w \mu( w )$; that is, if they are not matched through $\mu$, firm $f$ wants to hire $w$, and worker $w$ prefers firm $f$ rather than $\mu(w)$.  A matching $\mu$ is \textbf{stable} if it is individually rational and it is not blocked by any firm-worker pair. The set of stable matchings for market $(W,F,\succ_W,C_F)$ is denoted by $\boldsymbol{\mathcal{S}}.$ A matching is \textbf{worker-quasi-stable} if it is individually rational and each firm-worker blocking pair $(f,w)$ satisfies that $\mu (w)=\emptyset.$\footnote{The notion of worker-quasi-stable matching in many-to-one models generalizes the notion of ``simple" matching in one-to-one models studied by \cite{sotomayor1996non}. A one-to-one matching is \emph{simple} if, in the case of a blocking pair $(f,w)$ exists,  $\mu(w)=\emptyset$.} Let $\boldsymbol{\mathcal{Q}}$ denote  the set of worker-quasi-stable matchings for market $(W,F,\succ_W,C_F).$ Notice that, for each market $(W,F,\succ_W,C_F)$,  the set $\mathcal{Q}$ is always non-empty since the empty matching in which each agent is unmatched belongs to this set.

\cite{blair1988lattice} defines a partial order over matchings in which a matching dominates another matching  if each firm wishes to keep the workers hired under the first one, even if all the workers hired under the second one are also available, and do not wish to hire any new worker. Formally, given two sets of workers $S,T\in 2^W$, we write  $\boldsymbol{S \succeq^B_f T}$ when $S=C_f\left(S \cup T\right)$. We also write: $\boldsymbol{S \succ^B_f T}$ when $S \succeq^B_f T$ and $S \neq T$. Furthermore, given two matchings $\mu$ and $\mu'$, we say that \textbf{$\boldsymbol{\mu$ weakly Blair-dominates $\mu'}$}, and  write    $\boldsymbol{\mu \succeq^B \mu'},$ when $\mu(f) \succeq^B_f \mu'(f)$ for each $f \in F.$ If $\mu \succeq^B \mu'$ and $\mu \neq \mu',$  we say that \textbf{$\boldsymbol{\mu$ Blair-dominates $\mu'}$} and  write  $\boldsymbol{\mu \succ^B \mu'}.$

\section{Lattice structure}\label{lattice}

In this section, we prove that the set of worker-quasi-stable matchings forms a lattice under the partial order  $\succeq^B$. Formally,

\begin{theorem}\label{teorema lattice}
The set of worker-quasi-stable matchings is a lattice under the partial order $\succeq^B$.
\end{theorem}
The rest of the section is devoted to proving this theorem. In order to do so, we need to construct the \textit{join} of two worker-quasi-stable matchings.\footnote{Given a partially ordered set $(\mathcal{L},\succeq)$, and two elements $x,y\in \mathcal{L}$, an element $z\in \mathcal{L}$ is an \textit{upper bound} of $x$ and $y$ if $z\succeq x$ and $z\succeq y$. An element $w\in \mathcal{L}$ is the \textit{join} of $x$ and $y$ if and only if (i) $w$ is an upper bound of $x$ and $y$, and (ii) $t\succeq w$ for each upper bound $t$ of $x$ and $y$. The definitions of \textit{lower bound} and \textit{meet} of $x$ and $y$ are dual and we omit them.}
Given two worker-quasi-stable matchings $\mu$ and $\mu',$ we define a function $\lambda_{\mu,\mu'}:F\cup W \to 2^{F\cup W}$ as follows:
\begin{enumerate}[(i)]
\item for each $f\in F,$ $\lambda_{\mu,\mu'} (f)=C_f\left(\mu (f)\cup \mu ^{\prime }(f)\right),$
\item for each $w \in W,$ $\lambda_{\mu,\mu'} (w)=\{f\in F : w\in \lambda_{\mu,\mu'}(f)\}.$
\end{enumerate}

\noindent  Notice that by item (i),  if $\mu(f)=\mu'(f)=\emptyset$, then $\lambda_{\mu,\mu'}(f)=\emptyset.$ Under $\lambda_{\mu,\mu'},$ (i) firms want to  hire the best subset of workers  among those hired by them in either matching, and  (ii) workers agree with the firms that want  to employ them.   The following lemma shows that $\lambda_{\mu,\mu'}$ is well-defined (i.e., it is a matching) and, furthermore, that it is worker-quasi-stable.

\begin{lemma}\label{lambda es wqs}
If $\mu$ and $\mu'$ are two worker-quasi-stable matchings, then $\lambda_{\mu,\mu'} $ is a worker-quasi-stable matching.
\end{lemma}
\begin{proof}Let  $\mu,\mu' \in \mathcal{Q}$. First, we show $\lambda_{\mu,\mu'} $ is a matching. By definition of $\lambda_{\mu,\mu'} $ we have  $\lambda_{\mu,\mu'}(w)\subset F$ for each $w\in W$ and $\lambda_{\mu,\mu'}(f)\subset W$ for each $f\in F$. For $\lambda_{\mu,\mu'}$ to be matching, it is necessary to show that $|\lambda_{\mu,\mu'}(w)|\leq 1$ for each $w\in W$.  Assume that there is $w\in W$ such that $|\lambda_{\mu,\mu'}(w)|>1$. Thus, there are  $f,f'\in F$ with $f\neq f'$ such that $w\in\lambda_{\mu,\mu'}(f)$ and $w \in \lambda_{\mu,\mu'}(f')$.   Given that $\mu$ and $\mu'$ are matchings, w.l.o.g., assume that $w\in\mu(f)$ and $w \in \mu'(f').$ Thus, $\mu(w)=f$ and $\mu'(w)=f'$.
Since $w\in  \lambda_{\mu,\mu'}(f)=C_f\left(\mu(f)\cup\mu'(f)\right)$ then, by substitutability, $w\in  C_f\left(\mu'(f)\cup\{w\}\right).$ As $\mu'\in \mathcal{Q}$ and $\mu'(w)=f',$ then $(f,w)$ is not a blocking pair for $\mu'$. Therefore, 
\begin{equation}\label{ecu 1 lambda matching}
 \mu'(w)=f'\succ_w f.
\end{equation} 
By analogous reasoning, $w\in  \lambda_{\mu,\mu'}(f')$ implies 
\begin{equation}\label{ecu 2 lambda matching}
\mu(w)= f\succ_w f'.
\end{equation}
By (\ref{ecu 1 lambda matching}) and (\ref{ecu 2 lambda matching}) we get a contradiction. Therefore,  $|\lambda_{\mu,\mu'}(w)|\leq 1$ and $\lambda_{\mu,\mu'}$ is a matching.

Second, we show that $\lambda_{\mu,\mu'} $ is an individually rational matching. By \eqref{propiedad de choice}, for any $f\in F$ and $S \subseteq W$, $C_{f}\left( C_{f}\left(
S\right) \right) =C_{f}\left( S\right).$ Thus, $C_{f}\left( \lambda
_{\mu ,\mu ^{\prime }}\left( f\right) \right) =C_{f}\left( C_{f}\left( \mu
\left( f\right) \cup \mu'\left( f\right) \right) \right)
=C_{f}\left( \mu \left( f\right) \cup \mu' \left( f\right) \right)
=\lambda _{\mu ,\mu'}\left( f\right),$ and $\lambda _{\mu ,\mu'}$ is not blocked by any firm.
By definition of $\lambda_{\mu,\mu'}$ , $w\in \mu(f)$ or $w\in \mu'(f)$. Thus, $f=\mu(w)$ or $f= \mu'(w)$. Since $\mu$ and $\mu'$ are individually rational matchings, $\mu(w)\succ_{w}  \emptyset$ and $ \mu'(w)\succ_{w}  \emptyset$. Therefore, $f\succ_{w}  \emptyset$ and $\lambda _{\mu ,\mu'}$ is not blocked by any worker. This implies that $\lambda _{\mu ,\mu'}$ is an individually rational matching.

Finally, we show that $\lambda_{\mu,\mu'} $ is worker-quasi-stable. Assume that $\lambda_{\mu,\mu'}$ is not a worker-quasi-stable matching. Then, there is a blocking pair $(f,w)$ for $\lambda_{\mu,\mu'}$   and  $\lambda_{\mu,\mu'} (w)\neq \emptyset.$  Assume, w.l.o.g., that $\lambda_{\mu,\mu'} (w)=\mu(w).$ Recall that,  since $(f,w)$ is a blocking pair for $\lambda_{\mu,\mu'},$ $w\in C_f \left( \lambda_{\mu,\mu'}(f)\cup\{w\}\right),$ and $f\succ_w \lambda_{\mu,\mu'}(w).$ By definition of $\lambda_{\mu,\mu'}$, 

$$
w\in C_f\left( C_f\left(\mu(f)\cup\mu'(f) \right)\cup\{w\}\right).
$$
By  \eqref{propiedad de choice},  $w\in C_f\left( \mu(f)\cup\mu'(f)\cup\{w\}\right).$ Since firms have substitutable choice functions, $w\in C_f\left( \mu(f)\cup\{w\}\right).$ Recall that $f\succ_w \lambda_{\mu,\mu'}(w)=\mu(w).$  Therefore, $(f,w)$ is also a blocking pair for $\mu.$ Since we assumed that $\mu(w)=\lambda_{\mu,\mu'}(w)\neq \emptyset$, we contradict the fact that $\mu$ is, by hypothesis, a worker-quasi-stable matching.
\end{proof}

Next, we show that $\lambda_{\mu,\mu'}$ is indeed the join of $\mu$ and $\mu'$ according to $\succeq ^B$.
Let  $\mu,\mu' \in \mathcal{Q}$ and $f\in F$. 

\begin{lemma}\label{lema min cota sup}
If $\mu$ and $\mu'$ are two worker-quasi-stable matchings, then $\lambda_{\mu,\mu'}$ is the join of $\mu$ and $\mu'$.
\end{lemma}
\begin{proof}Let  $\mu,\mu' \in \mathcal{Q}.$ We know by Lemma \ref{lambda es wqs}, that $\lambda_{\mu,\mu'}\in \mathcal{Q}.$   
First, we prove that $\lambda_{\mu,\mu'}$ is an upper bound of $\mu$ and $\mu'$.
By definition of $\lambda_{\mu,\mu'}$ and  \eqref{propiedad de choice}, for each $f\in F$,
$$C_f \left(\lambda_{\mu,\mu'}(f) \cup \mu(f) \right)=C_f \left(C_f\left(\mu(f) \cup \mu'(f)\right)\cup\mu(f) \right)$$
$$=C_f \left(\mu(f) \cup \mu'(f)\cup \mu(f)\right)=C_f \left(\mu(f) \cup \mu'(f) \right)=\lambda_{\mu,\mu'}(f).$$
This implies that $\lambda_{\mu,\mu'}(f) \succeq_f^B \mu(f)$ for each $f\in F$ and then $\lambda_{\mu,\mu'} \succeq ^B \mu$.  Similarly, $\lambda_{\mu,\mu'} \succeq ^B \mu'$.
Therefore, $\lambda_{\mu,\mu'}$ is an upper bound of  $\mu$ and $\mu'$. Second,  we prove that $\lambda_{\mu,\mu'}$ is the join of $\mu$ and $\mu'$. Let $\nu\in \mathcal{Q}$ be  such that $\nu \succeq^B \mu$ and $\nu \succeq^B  \mu' $. That is, \begin{equation}\label{ecu1 min cot sup}
\nu(f)=C_f \left(\nu(f) \cup \mu(f) \right)\text{~~~and~~~}
\nu(f)=C_f \left(\nu(f) \cup \mu'(f) \right)
\end{equation}
for each $f\in F.$ 
We need to show that $\nu\succeq ^B \lambda_{\mu,\mu'}$, that is, $\nu(f)=C_f \left(\nu(f) \cup \lambda_{\mu,\mu'}(f)) \right)$ for each $f\in F.$ Thus, using repeatedly \eqref{propiedad de choice} and \eqref{ecu1 min cot sup}, and the definition of $\lambda_{\mu,\mu'}$,
$$\nu(f)=C_f \left(\nu(f) \cup \mu(f) \right)
=C_f \left(C_f \left(\nu(f) \cup \mu'(f) \right) \cup \mu(f) \right)$$
$$=C_f \left(\nu(f) \cup \mu'(f) \cup \mu(f) \right)
=C_f \left(\nu(f) \cup C_f \left(\mu'(f)  \cup \mu(f)\right) \right)=C_f \left(\nu(f) \cup \lambda_{\mu,\mu'}(f)) \right)$$ for each $f\in F.$
Thus, $\nu\succeq ^B \lambda_{\mu,\mu'}.$  Therefore, $\lambda_{\mu,\mu'}$ is the join for $\mu$ and $\mu'$.
\end{proof}

From now on, given two worker-quasi-stable matchings $\mu$ and $\mu'$, we denote $\lambda_{\mu,\mu'}$ as $\mu \vee\mu'.$

To finish the proof of Theorem \ref{teorema lattice}, we make three observations. First, by Lemma \ref{lema min cota sup} the set of worker-quasi-stable matchings $\mathcal{Q}$ forms a join-semilattice under the partial order $\succeq^B$.\footnote{A partially ordered set $\mathcal{L}$ is called a \textit{join-semilattice} if any two elements in $\mathcal{L}$ have a join. If any two elements in  $\mathcal{L}$ also have a meet, then $\mathcal{L}$ is called a \textit{lattice} \citep[see][for more details]{stanley1986enumerative}.} Second, the  empty matching $\mu_{\emptyset}$ in which all workers are unmatched (that is by definition a worker-quasi-stable matching) is the minimum element of $\mathcal{Q}$ under the partial order $\succeq^B$. To see this, let $\mu\in\mathcal{Q}$. Since $\mu(f)=C_f(\mu(f) \cup \emptyset)$ for each $f\in F$, it follows that $\mu=\mu \vee \mu_{\emptyset}$. Thus, $\mu \succeq^B \mu_{\emptyset}$ for each $\mu \in \mathcal{Q}.$ 
Finally, given that the set of worker-quasi-stable matchings is finite and is a join-semilattice with a minimum element, it follows that the set of worker-quasi-stable matchings forms a lattice under the partial order $\succeq^B$  \citep[see][for more details]{stanley1986enumerative}. This completes the proof of Theorem \ref{teorema lattice}.

The following example illustrates the lattice structure of the set of worker-quasi-stable matchings. 
\begin{example}\label{ejemplo 1}
Let $(W,F,\succ_W,C_F)$ be a matching market where $W=\{w_1,w_2,w_3,w_4\}$,  $F=\{f_1,f_2\}$,  the preference for the workers are given by:

\noindent\begin{tabular}{l}
$\succ_{w_i}:\{f_1\},\{f_2\},\emptyset~~~~ \text{ for }i=1,2$\\
$\succ_{w_3}:\{f_2\},\{f_1\},\emptyset$\\
$\succ_{w_4}:\{f_2\},\emptyset,$\\
\end{tabular}

\noindent and the choice functions of the firms are given in Table \ref{tabla ejemplo 1}.

\begin{table}[h!]
\begin{center}
\begin{tabular}{|c|c|c|c|c|c|c|c|c|c|c|c|c|c|c|c|} \hline
&$\boldsymbol{1234}$&$\boldsymbol{123}$&$\boldsymbol{234}$&$\boldsymbol{134}$&$\boldsymbol{124}$&$\boldsymbol{12}$&$\boldsymbol{13}$&$\boldsymbol{14}$&$\boldsymbol{23}$&$\boldsymbol{24}$&$\boldsymbol{34}$&$\boldsymbol{1}$&$\boldsymbol{2}$&$\boldsymbol{3}$&$\boldsymbol{4}$  \\ \hline \hline
$\boldsymbol{C_{f_1}}$&$3$&$3$&$3$&$3$&$12$&$12$&$3$&$1$&$3$&$2$&$3$&$1$&$2$&$3$&$\emptyset$  \\ \hline
$\boldsymbol{C_{f_2}}$&$12$&$12$&$24$&$13$&$12$&$12$&$13$&$1$&$2$&$24$&$34$&$1$&$2$&$3$&$4$  \\ \hline
\end{tabular}\label{tabla ejemplo 1}
\caption{Choice functions of the firms}
\end{center}
\end{table}

\noindent For example, $C_{f_1}(\{w_1,w_2,w_3,w_4\})=\{w_3\}$.
We denote each worker-quasi-stable matching as an ordered pair, in which the first component consists of the workers hired by $f_1$, and the second of the workers hired by $f_2$.
In Figure \ref{reticulado1},  the lattice of all  worker-quasi-stable matchings is presented.
\begin{figure}[h!]
\centering
\begin{tikzpicture}[scale=1.1]
\node (0) at (0,0) {\small{$(\emptyset,\emptyset)$}};
\node (1) at (-3,2) {\small{$(w_1,\emptyset)$}};
\node (2) at (-1,2) {\small{$(w_2,\emptyset)$}};
\node (3) at (1,2) {\small{$(\emptyset,w_3)$}};
\node (4) at (3,2) {\small{$(\emptyset,w_4)$}};
\node (5) at (-5,4) {\small{$(w_1w_2,\emptyset)$}};
\node (6) at (-3,4) {\small{$(\emptyset,w_3w_4)$}};
\node (7) at (-1,4) {\small{$(w_1,w_3)$}};
\node (8) at (1,4) {\small{$(w_1,w_4)$}};
\node (9) at (3,4) {\small{$(w_2,w_3)$}};
\node (10) at (5,4) {\small{$(w_2,w_4)$}};
\node (11) at (-3,6) {\small{$(w_1w_2,w_3)$}};
\node (12) at (-1,6) {\small{$(w_1w_2,w_4)$}};
\node (13) at (1,6) {\small{$(w_1,w_3w_4)$}};
\node (14) at (3,6) {\small{$(w_2,w_3w_4)$}};
\node (15) at (-1,8) {\small{$(w_3,w_2)$}};
\node (16) at (1,8) {\small{$\boldsymbol{(w_1w_2,w_3w_4)}$}};
\node (17) at (0,10) {\small{$(w_3,w_2w_4)$}};
\node (18) at (0,12) {\small{$\boldsymbol{(w_3,w_1w_2)}$}};
\draw (0) to (1);
\draw (0) to (2);
\draw (0) to (3);
\draw (0) to (4);

\draw (1) to (5);
\draw (1) to (7);
\draw (1) to (8);
\draw (2) to (5);
\draw (2) to (9);
\draw (2) to (10);
\draw (3) to (6);
\draw (3) to (7);
\draw (3) to (9);
\draw (4) to (6);
\draw (4) to (8);
\draw (4) to (10);

\draw (5) to (11);
\draw (5) to (12);
\draw (6) to (13);
\draw (6) to (14);
\draw (7) to (11);
\draw (7) to (13);
\draw (8) to (12);
\draw (8) to (13);
\draw (9) to (11);
\draw (9) to (14);
\draw (10) to (12);
\draw (10) to (14);

\draw (11) to (15);
\draw (11) to (16);
\draw (12) to (16);
\draw (13) to (16);
\draw (14) to (16);

\draw (15) to (17);
\draw (16) to (17);

\draw (17) to (18);
\end{tikzpicture}
\caption{The lattice for Example \ref{ejemplo 1}.}
\label{reticulado1}
\end{figure}
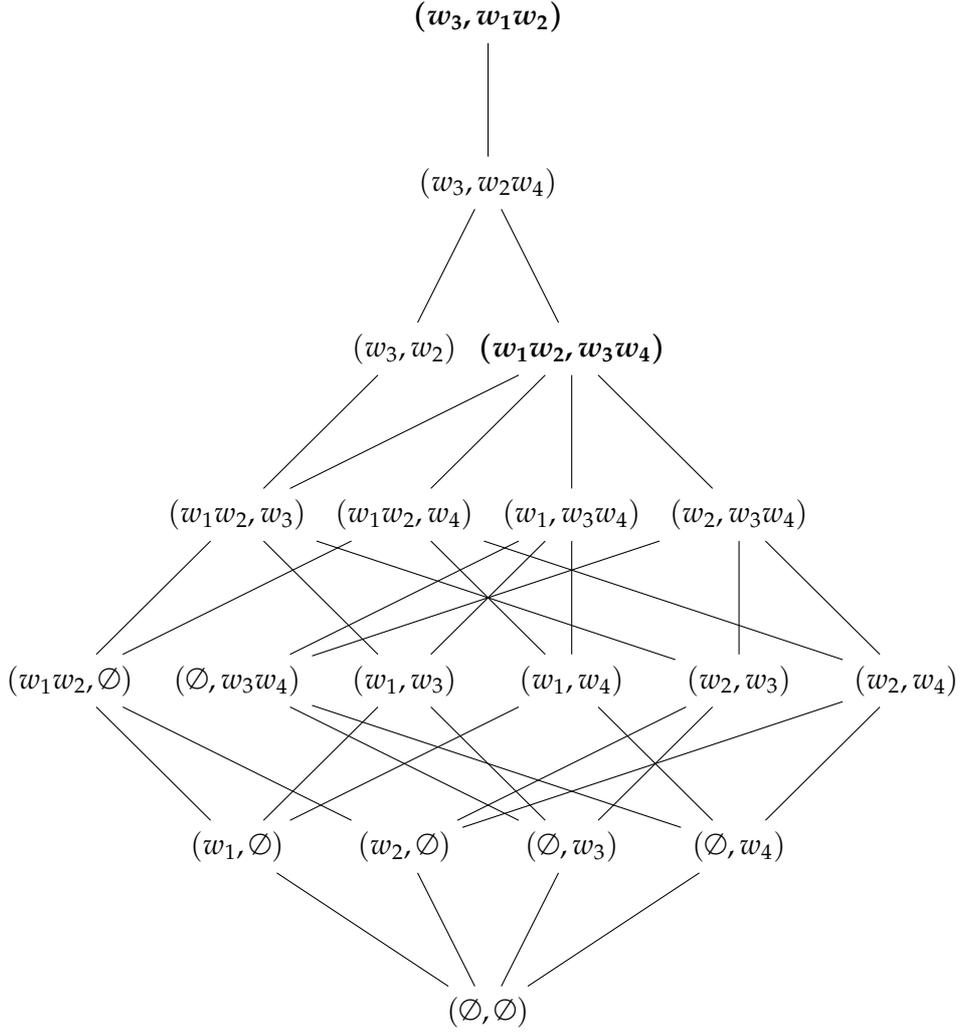
\noindent There are nineteen worker-quasi-stable matchings, two of which are stable: $(w_1w_2,w_3w_4)$ and $(w_3,w_1w_2).$ Now, we illustrate how to compute the join of two  worker-quasi-stable matchings. Take, for instance, $\mu'=(w_3,w_2)$ and $\overline{\mu}=(w_1w_2,w_3w_4).$  Observe that $$C_{f_1}\left(\mu'\vee\overline{\mu}(f_1)\right)=C_{f_1}\left(\{w_3\} \cup \{w_1,w_2\}\right)=\{w_3\}$$ and $$C_{f_2}\left(\mu'\vee\overline{\mu}(f_2)\right)=C_{f_2}\left(\{w_2\} \cup \{w_3,w_4\}\right)=\{w_2,w_4\}.$$
Then, $\mu'\vee\overline{\mu}=(w_3,w_2w_4).$ \hfill $\Diamond$
\end{example}

To finish this section, we compare our approach to construct a join with the one of \cite{wu2018lattice}. In their  Example 5.1, for two worker-quasi-stable-matchings, the authors  show the difficulties in defining a Conway-like join for firms even when firms have responsive preferences.\footnote{For a firm $f\in F$, preference $\succ_f$  is \textbf{responsive} if there is a quota $q_f$ such that, for each $S\subseteq W$:
\begin{enumerate}[(i)]
\item If $\vert{S}\vert>q_{f}$, we have $\emptyset\succ_{f}S$.
\item If $\vert{S}\vert<q_{f}$ and $w\not\in{S}$, we have $S\cup\{w\}\succ_{f}S \mbox{ if and only if } w\succ_{f}\emptyset.$
\item If $\vert S  \vert< q_{f}$ and $w,w'\not\in{S}$, we have 
$S \cup \{w\}\succ_{f} S\cup \{w'\} \mbox{ if and only if } w\succ_{f}w'.$
\end{enumerate} } In their example, there is only one firm $f$ with quota $q_f=3$ and the set of workers is $\{w_1,w_2,w_3,w_4\}$.  All workers rank $f$ above $\emptyset$ and  the ranking of single agents of $f$ is $w_1\succ_fw_2\succ_fw_3\succ_fw_4\succ_f \emptyset.$
Now, take the  worker-quasi-stable matchings $\mu=(w_1w_4)$ and $\mu'=(w_2w_3)$. Notice that the sets of agents assigned to $f$ in these two matchings are not comparable in a responsive manner, therefore the  Conway-like join can not be defined.
Here in this paper, we are able to sidestep this problem following Blair's insight. 
Now, we show how our results apply to their example.
According to our approach, we can compute the join of $\mu$ and $\mu'$ as follows:
$$\mu \vee \mu'(f)=C_{f}\left(\mu(f)\cup \mu'(f)\right)=C_{f}\left(\{w_1,w_2\} \cup \{w_2,w_3\}\right)=\{w_1,w_2,w_3\}.$$
Therefore, the resulting  worker-quasi-stable matching is $\mu\vee \mu'=(w_1w_2w_3).$ Furthermore, this matching is stable.

\section{Re-equilibration process via a  Tarski operator} \label{seccion operador }
As we mentioned in the introduction,  worker-quasi-stable matchings appear naturally in the short run when firms decide to downsize  or  new workers become available.
In this section, we define our Tarski operator in the worker-quasi-stable lattice that describes a possible re-equilibration process.  This process models how, starting from a worker-quasi-stable matching,  a decentralized sequence of offers in which unemployed workers are hired and cause new unemployments, produces a sequence of worker-quasi-stable matchings that converges to a stable matching.  In the first subsection, we present the operator, show some of its properties, and prove that the set of  its fixed points is the set of stable matchings. In the second subsection, we discuss the re-equilibration process, based on our Tarski operator, that models a layoff chain that leads towards a stable matching. 

\subsection{A Tarski operator for worker-quasi-stable matchings}\label{subsection operator T}

First, for each $\mu\in \mathcal{Q},$ define the following sets:$$\mathcal{B}^{\mu}=\{(f,w) \in F \times W : (f, w) \text{ blocks } \mu\},$$
$$\mathcal{B}^{\mu}_\star=\left\{(f,w) \in  \mathcal{B}^{\mu} : f \succ_w f' \text{ for each } (f',w) \in \mathcal{B}^{\mu}\setminus \{(f,w)\}\right\}$$
\noindent and, for each $f\in F$,
$$\mathcal{W}^\mu_f=\{ w \in W : (f,w) \in \mathcal{B}^{\mu}_\star\}.$$

\noindent The set $\mathcal{B}^{\mu}$ collects all possible blocking pairs for $\mu$ (since $\mu \in \mathcal{Q},$ $(f,w)\in \mathcal{B}^{\mu}$ implies  $\mu(w)=\emptyset$). Then,  each blocking pair $(f,w)$ is
included in $\mathcal{B}^{\mu}_{\star}$ if $f$ is $w$'s most preferred firm with which $w$ forms a blocking pair.  Lastly, for each $f\in F$, $\mathcal{W}^\mu_f$ gathers all workers that partner with $f$ in $\mathcal{B}_{\star}^{\mu}$. Notice that, if a firm $f$ is not involved in any blocking pair in $\mathcal{B}_{\star}^{\mu}$, then  $\mathcal{W}^\mu_f=\emptyset.$ Now, for each $\mu \in \mathcal{Q},$ our Tarski operator  $T$ maps $\mathcal{Q}$ into the set of matchings and  is defined as follows:  
\begin{enumerate}[(i)]
\item for each $f\in F$, $T(\mu)(f)=C_f \left(\mu(f) \cup \mathcal{W}^\mu_f\right)$
\item for each $w\in W,$ if there is $f\in F$ such that $w \in T(\mu)(f)$ then $T(\mu)(w)=f.$ Otherwise, $T(\mu)(w)=\emptyset.$
\end{enumerate}

\begin{remark}\label{remark T matching}
If  $w\in \mathcal{W}^\mu_f$, (i) $\mu(w)= \emptyset$ and, (ii) by definition of $\mathcal{B}^{\mu}_\star$,  $w\notin \mathcal{W}^\mu_{f'}$ for each $f'\in F\setminus \{f\}$.
\end{remark}
The following lemma shows that the operator $T$ is well defined: it assigns to a worker-quasi-stable matching  a matching.
 \begin{lemma}\label{T de mu es matching} 
For any worker-quasi-stable matching $\mu$,   $T(\mu)$ is a matching.
 \end{lemma} 
 \begin{proof}
 Let $\mu$ be a worker-quasi-stable matching. By definition of $T( \mu)$, $T(\mu)(f)\subseteq W$ for each $f\in F$ and $T(\mu)(w)\subseteq F$ for each $w\in W$. To show that $T(\mu)$ is a matching, it only remains to be  seen that $|\{f\in F: T(\mu)(w)=f\}|\leq1$ for each $w\in W$. Thus, assume that there are $w\in W$ and distinct $f$ and $f'$ in $F$ such that $w\in T(\mu)(f)$ and $w\in T(\mu)(f')$. By definition of $T(\mu)$ and choice function,   
\begin{equation}\label{eq1 T matching}
w\in C_f \left(\mu(f) \cup \mathcal{W}^\mu_f \right)\subseteq \mu(f) \cup \mathcal{W}^\mu_f.
\end{equation}
 Likewise,  \begin{equation}\label{eq2 T matching}
 w\in C_{f'} \left(\mu(f') \cup \mathcal{W}^\mu_{f'} \right)\subseteq \mu(f') \cup \mathcal{W}^\mu_{f'}.
 \end{equation}
 By \eqref{eq1 T matching}, there are two cases to consider:
\begin{enumerate}
\item[$\boldsymbol{1}.$] $\boldsymbol{w\in \mu(f)}.$ Since $\mu$ is a matching, $w\notin \mu(f')$. By \eqref{eq2 T matching}, $w\in \mathcal{W}^\mu_{f'}$ and, therefore, by Remark \ref{remark T matching} (i) $\mu(w)=\emptyset$. This contradicts that $w\in \mu(f)$. 
\item[$\boldsymbol{2}.$]$\boldsymbol{w\in \mathcal{W}^\mu_{f}}.$ By Remark \ref{remark T matching} (ii) $w\notin \mathcal{W}^\mu_{f'}$. By \eqref{eq2 T matching}, $w\in \mu(f')$. Thus, $\mu(w)\neq \emptyset$. This contradicts Remark \ref{remark T matching} (i).
\end{enumerate}
\end{proof}

The following theorem states that the  matching obtained by applying our operator to a worker-quasi-stable matching is: 
(i) worker-quasi-stable, (ii) weakly Blair-preferred by firms  to the original matching, and (iii) identical
to the original matching if and only if the original matching is stable.

\begin{theorem}  \label{Proposition punto fijo}
For any worker-quasi-stable matching $\mu,$ the following hold: 
\begin{enumerate}[(i)]
\item $T(\mu)$ is a worker-quasi-stable matching,
\item $T(\mu)\succeq ^B \mu$,
\item $T(\mu)=\mu$ if and only if $\mu$ is stable.

\end{enumerate}
\end{theorem}
\begin{proof} Let $\mu \in \mathcal{Q}.$ 
\begin{enumerate}[(i)]
\item \textbf{$\boldsymbol{T(\mu)}$ is a worker-quasi-stable matching.} 
Let $(f,w)$ be a blocking pair of $T(\mu)$. We want to see that $T(\mu)(w)=\emptyset.$ Assume there is $f'\in F$ such that $T(\mu)(w)=f'$. Then, by definition of $T(\mu)(f')$, either $w \in \mu(f')$ or $w \in \mathcal{W}_{f'}^\mu.$ 
\begin{enumerate}
\item[$\boldsymbol{1}.$] $\boldsymbol{w \in \mu(f')}$. Since $(f,w)$ blocks $T(\mu)$, $w\in C_f (T(\mu)(f)\cup \{w\}).$ By definition of $T(\mu)$ and \eqref{propiedad de choice}, $w\in C_f (\mu(f)\cup \mathcal{W}_{f}^\mu \cup \{w\})$. By substitutability,   $w\in C_f (\mu(f) \cup \{w\}).$
Also since $(f,w)$ blocks $T(\mu)$, $f\succ_w T(\mu)(w)=\mu(w)$. Thus $(f,w)$ blocks $\mu$ and $\mu(w)\neq \emptyset$, contradicting that $\mu\in \mathcal{Q}.$
\item[$\boldsymbol{2}.$] $\boldsymbol{w \in \mathcal{W}_{f'}^\mu}$. Thus, $(f',w) \in \mathcal{B}_\star^\mu$ and this implies that there is no firm $f'' \in F\setminus \{f'\}$ such that $f'' \succ_w f'$  and $(f'', w)$  blocks $T(\mu).$ In particular, when  $f''=f$, $(f,w)$ cannot be a blocking pair of $T(\mu)$,  contradicting our assumption.
\end{enumerate}
Hence,   $T(\mu)(w)=\emptyset$ and, therefore, $T(\mu)$ is a worker-quasi-stable matching. 

\item $\boldsymbol{T(\mu)\succeq ^B \mu.}$ By definition of $T(\mu)$ and  \eqref{propiedad de choice}, for each $f\in F,$
$$
C_f\left(T(\mu)(f)\cup \mu(f)\right)=C_f\left(C_f\left(\mu(f)\cup \mathcal{W}^\mu_f\right)\cup \mu(f)\right)$$
$$=C_f\left(\mu(f)\cup \mathcal{W}^\mu_f \cup \mu(f)\right)=C_f\left(\mu(f)\cup \mathcal{W}^\mu_f\right)=T(\mu)(f).$$ Therefore, $C_f\left(T(\mu)(f)\cup \mu(f)\right)=T(\mu)(f)$ for each $f\in F$, as desired.

\item $\boldsymbol{T(\mu)=\mu}$ \textbf{if and only if $\boldsymbol{\mu}$ is stable.}  Assume $\mu\in \mathcal{Q}\setminus \mathcal{S}$. Thus, $\mathcal{B}^{\mu}_\star \neq \emptyset$ and, therefore, there is $f\in F$  such that $\mathcal{W}^\mu_f \neq \emptyset$.  This implies that $T(\mu)(f)= C_f\left(\mu(f)\cup \mathcal{W}^\mu_f\right) \neq \mu(f).$ Hence, $T(\mu)\neq \mu.$

\noindent  Assume that $\mu\in \mathcal{S}$. Thus, $\mathcal{B}^{\mu}=\mathcal{B}^{\mu}_\star=\emptyset.$ By definition of $T,$ $T(\mu)(f)=\mu(f)$ for each $f\in F$ and therefore  $T(\mu)=\mu$.\end{enumerate}
\end{proof}

\noindent Notice that Theorem \ref{Proposition punto fijo} implies that $T$ is a Pareto improving operator for the firms by definition of the choice function.  Now, we prove that our operator $T$ is isotone. Recall that, for a lattice  $(\mathcal{L},\geq )$, a function $T:\mathcal{L}\longrightarrow\mathcal{L}$ is  \textit{isotone} if for each $x,y \in \mathcal{L}$, $x \geq y$ implies $T(x) \geq T(y).$

\begin{lemma}\label{T isotone}
If $\mu$ and $\mu'$ are worker-quasi-stable matchings such that $\mu \succeq^B  \mu'$, then $T(\mu) \succeq^B  T(\mu').$
\end{lemma}
\begin{proof}
Let $\mu, \mu' \in \mathcal{Q}$ be such that $\mu \succeq^B  \mu'$ and assume that $T(\mu) \succeq^B  T(\mu')$ does not hold. This implies the existence of $f \in F$ such that \begin{equation}\label{T1}
T(\mu)(f) \neq C_f \left(T(\mu)(f) \cup T(\mu')(f) \right). 
\end{equation}
Using the definition of $T$ and  \eqref{propiedad de choice} twice, it follows that 
%\begin{align*}
\begin{equation}\label{t1}
 C_f\left( T(\mu)(f) \cup T(\mu')(f)\right)=C_f\left(C_f \left(\mu(f)\cup \mathcal{W}^\mu_f \right) \cup C_f \left(\mu'(f)\cup \mathcal{W}^{\mu'}_f \right)\right)
\end{equation}
$$ =C_f\left(\mu(f)\cup \mathcal{W}^\mu_f  \cup C_f \left(\mu'(f)\cup \mathcal{W}^{\mu'}_f \right)\right)=C_f \left( \mu(f) \cup \mathcal{W}^\mu_f \cup \mu'(f)  \cup \mathcal{W}^{\mu'}_f \right) $$
$$=C_f \left( \left(\mu(f)\cup \mu'(f) \right) \cup \mathcal{W}^\mu_f \cup \mathcal{W}^{\mu'}_f \right) .$$
Using again \eqref{propiedad de choice}, it follows that 
\begin{equation}\label{t2}
 C_f \left( \left(\mu(f)\cup \mu'(f) \right) \cup \mathcal{W}^\mu_f \cup \mathcal{W}^{\mu'}_f \right)=C_f \left(C_f \left( \mu(f)\cup \mu'(f) \right) \cup \mathcal{W}^\mu_f \cup \mathcal{W}^{\mu'}_f \right)
\end{equation}
%\end{align*}
and, as by hypothesis $C_f(\mu(f) \cup \mu'(f))=\mu(f),$ using \eqref{t1} and \eqref{t2} we get 
\begin{equation}\label{T2}
C_f\left(T(\mu)(f) \cup T(\mu')(f)\right)=C_f \left( \mu(f) \cup \mathcal{W}^\mu_f \cup \mathcal{W}^{\mu'}_f\right).
\end{equation}
 Now, using the definition of $T$ and \eqref{T2},    \eqref{T1} becomes
 \begin{equation}\label{T3}
C_f\left( \mu(f) \cup \mathcal{W}^\mu_f \right)  \neq C_f\left( \mu(f) \cup \mathcal{W}^\mu_f \cup \mathcal{W}^{\mu'}_f \right).
 \end{equation}
 By \eqref{T3}, there is $w \in W$ such that 
\begin{equation}\label{T4}
w \in C_f\left( \mu(f) \cup \mathcal{W}^\mu_f \cup \mathcal{W}^{\mu'}_f \right)
\end{equation}
and 

\begin{equation}\label{T5}
w \in \mathcal{W}^{\mu'}_f \setminus \left(\mu(f) \cup \mathcal{W}^\mu_f \right).
\end{equation}
Notice that $w \in \mathcal{W}^{\mu'}_f$  implies $\mu'(w)=\emptyset.$ There are two cases to consider:
\begin{enumerate}
\item[$\boldsymbol{1}.$] $\boldsymbol{\mu(w)=\emptyset}.$ Since $w \in \mathcal{W}^{\mu'}_f,$ $(f,w) \in \mathcal{B}^{\mu'}_\star$ so $f \succ_w \mu'(w)=\emptyset.$  By \eqref{T4} and substitutability, $w \in C_f\left(\mu(f) \cup \{w\} \right),$ so $(f,w) \in \mathcal{B}^\mu.$ By \eqref{T5}, $w \notin \mathcal{W}^\mu_f$ and hence there is $f' \in F$ such that $(f',w) \in \mathcal{B}^\mu$ and $f' \succ_w f.$ Then, 
\begin{equation}\label{T6}
w \in C_{f'} \left( \mu(f')\cup \{w\}\right).
\end{equation}
By hypothesis, $\mu(f') \succeq^B_{f'} \mu'(f')$ and therefore $\mu(f')=C_{f'} \left( \mu(f') \cup \mu'(f')\right).$ Thus, using  \eqref{propiedad de choice} we can rewrite \eqref{T6} as 
\begin{equation}\label{T7}
w \in C_{f'} \left( \mu(f') \cup \mu'(f') \cup \{w\}\right).
\end{equation}
Hence, substitutability and \eqref{T7} imply $w \in C_{f'} \left(\mu'(f')\cup \{w\}\right)$ and therefore $(f',w) \in \mathcal{B}^{\mu'}.$ But $f' \succ_w f$ contradicts the fact that $w \in \mathcal{W}^{\mu'}_f.$      

\item[$\boldsymbol{2}.$] \textbf{There is $\boldsymbol{f' \in F$ such that $\mu(w)=f'}$}. By \eqref{T5}, $f'\neq f.$ Notice that, as $\mu \succeq^B  \mu'$,  $w \in \mu(f')=C_{f'}(\mu(f'))=C_{f'}  \left(\mu(f')\cup \mu'(f')\right)$ implies, by substitutability, that $w \in C_{f'}  \left(\mu'(f')\cup \{w\}\right).$ By \eqref{T5}, $w \in \mathcal{W}^{\mu'}_f$ and thus $(f,w) \in \mathcal{B}^{\mu'}_\star.$ If $f' \succ_w f,$ then $(f',w) \in \mathcal{B}^{\mu'}$ but this contradicts the fact  that $(f,w) \in \mathcal{B}^{\mu'}_\star.$ Therefore, 
\begin{equation}\label{T8}
f \succ_w f'=\mu(w) \succ_w \emptyset.
\end{equation}
By \eqref{T4} and substitutability, we have that $w \in C_f\left(\mu(f) \cup \{w\} \right).$ Then, \eqref{T8} implies that $(f,w) \in \mathcal{B}^\mu,$ 
contradicting the fact that $\mu$ is worker-quasi-stable since $\mu(w)=f'\neq \emptyset.$ 
\end{enumerate}
Since in each case we reach a contradiction, it follows that $$T(\mu)(f)=C_f \left(T(\mu)(f) \cup T(\mu')(f) \right)$$ for each $f \in F,$ which in turn implies  $T(\mu) \succeq^B  T(\mu').$
\end{proof}

Starting from any $\mu\in \mathcal{Q}$, the following example illustrates the construction of the sets $\mathcal{B}^{\mu}$, $\mathcal{B}^{\mu}_\star$, and $\mathcal{W}_{f}^{\mu}$  for each $f\in F$, in order to compute our operator $T$. Moreover, applying our operator $T$ (sometimes more than once) we obtain the set of fixed points of the  lattice of the set of  worker-quasi-stable matchings. In this particular example, the lattice of fixed points consists of the unique stable matching of the market.

\begin{example}\label{ejemplo 2}
Let $(W,F,\succ_W,C_F)$ be a matching market where  $W=\{w_1,w_2,w_3\}$, $F=\{f_1,f_2\}$, and the preferences for the workers are  is given by: $$\succ_{w_i}:\{f_1\},\{f_2\},\emptyset~~~~ \text{ for }i=1,2,3.$$ Moreover, the choice functions of the firms are given in Table \ref{tabla ejemplo 2}.
\begin{table}[h!]
\begin{center}
\begin{tabular}{|c|c|c|c|c|c|c|c|} \hline
&$\boldsymbol{123}$&$\boldsymbol{12}$&$\boldsymbol{13}$&$\boldsymbol{23}$&$\boldsymbol{1}$&$\boldsymbol{2}$&$\boldsymbol{3}$  \\ \hline \hline
$\boldsymbol{C_{f_1}}$&$12$&$12$&$13$&$23$&$1$&$2$&$3$  \\ \hline
$\boldsymbol{C_{f_2}}$&$3$&$\emptyset$&$3$&$3$&$\emptyset$&$\emptyset$&$3$  \\ \hline
\end{tabular}\label{tabla ejemplo 2}
\caption{Choice functions of the firms}
\end{center}
\end{table}
In Figure \ref{reticulado2}, we present the lattice of the set of worker-quasi-stable matchings.
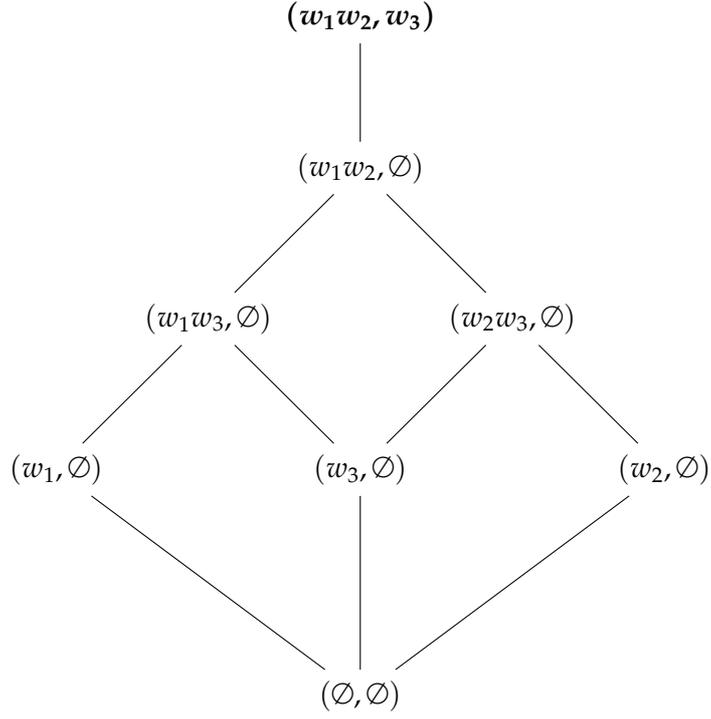
\begin{figure}[h!]
\begin{center}
\begin{tikzpicture}[scale=1]

\node (5) at (-4,-2) {\small{$(w_1,\emptyset)$}};
\node (6) at (4,-2) {\small{$(w_2,\emptyset)$}};
\node (7) at (0,-5) {\small{$(\emptyset,\emptyset)$}};

\node (0) at (0,-2) {\small{$(w_3,\emptyset)$}};
\node (1) at (-2,0) {\small{$(w_1w_3,\emptyset)$}};
\node (2) at (2,0) {\small{$(w_2w_3,\emptyset)$}};
\node (3) at (0,2) {\small{$(w_1w_2,\emptyset)$}};
\node (4) at (0,4) {\small{$\boldsymbol{(w_1w_2,w_3)}$}};

\draw (0) to (1);
\draw (0) to (2);
\draw (1) to (3);
\draw (2) to (3);

\draw (3) to (4);

\draw (7) to (5);
\draw (7) to (6);
\draw (7) to (0);
\draw (5) to (1);
\draw (6) to (2);

\end{tikzpicture}
\caption{The lattice of Example \ref{ejemplo 2}.}
\label{reticulado2}
\end{center}
\end{figure}
There are eight worker-quasi-stable matchings, one of which is stable: $\overline{\mu}=(w_1w_2,w_3).$
Now, we show how our operator $T$ works. Start, for instance, with  $\mu'=(w_2w_3,\emptyset).$
First, we construct the sets $\mathcal{B}^{\mu'}$, $\mathcal{B}^{\mu'}_\star$, $\mathcal{W}_{f_1}^{\mu'}$ and $\mathcal{W}_{f_2}^{\mu'}$ and compute $T(\mu'):$
$$\mathcal{B}^{\mu'}=\mathcal{B}^{\mu'}_\star=\{(f_1,w_1)\},~ \mathcal{W}_{f_1}^{\mu'}=\{w_1\} \text{ and }\mathcal{W}_{f_2}^{\mu'}=\emptyset,$$
$$T(\mu')(f_1)=C_{f_1}\left(\mu'(f_1) \cup \mathcal{W}_{f_1}^{\mu'}\right)=C_{f_1}\left(\{w_2,w_3\} \cup \{w_1\}\right)=\{w_1,w_2\},$$
$$T(\mu')(f_2)=C_{f_2}\left(\mu'(f_2) \cup \mathcal{W}_{f_2}^{\mu'}\right)=C_{f_2}\left(\emptyset\cup\emptyset \right)=\emptyset.$$
\noindent Therefore, $T(\mu')=\widetilde{\mu}=(w_1w_2,\emptyset).$ Notice that  matching $\widetilde{\mu}$ is not stable, since $(f_2,w_3)$ is a blocking pair.
Second, we construct the sets $\mathcal{B}^{\widetilde{\mu}}$, $\mathcal{B}^{\widetilde{\mu}}_\star$, $\mathcal{W}_{f_1}^{\widetilde{\mu}}$ and $\mathcal{W}_{f_2}^{\widetilde{\mu}}$ and compute $T^2(\mu')\equiv T(\widetilde{\mu}):$
$$\mathcal{B}^{\widetilde{\mu}}=\mathcal{B}^{\widetilde{\mu}}_\star=\{(f_2,w_3)\},~ \mathcal{W}_{f_1}^{\widetilde{\mu}}=\emptyset\text{ and }\mathcal{W}_{f_2}^{\widetilde{\mu}}=\{w_3\}, $$
$$T(\widetilde{\mu})(f_1)=C_{f_1}\left((\widetilde{\mu}(f_1) \cup \mathcal{W}_{f_1}^{\widetilde{\mu}}\right)=C_{f_1}\left(\{w_1,w_2\} \cup \emptyset\right)=\{w_1,w_2\},$$
$$T(\widetilde{\mu})(f_2)=C_{f_2}\left(\widetilde{\mu}(f_2) \cup \mathcal{W}_{f_2}^{\widetilde{\mu}}\right)=C_{f_2}\left(\emptyset\cup\{w_3\} \right)=\{w_3\}.$$
\noindent Therefore, $T(\widetilde{\mu})=(w_1w_2,w_3)=\overline{\mu}.$ It can be shown that, for each $\mu\in \mathcal{Q}$ such that $\widetilde{\mu} \succeq^B  \mu$, $T(\mu)=\widetilde{\mu}$ so $T^2(\mu)=\overline{\mu}.$ Furthermore, we can observe that the fixed point of our operator $T$ is the stable matching $\overline{\mu}.$
 \hfill $\Diamond$  
\end{example}

Note that in Example \ref{ejemplo 2}, the set of fixed points of operator $T$ consists of the unique stable matching of the market and we cannot observe its lattice structure. However, in Example \ref{ejemplo 1}, it can be shown that the set of fixed points of  $T$ is the set of stable matchings $\{(w_1w_2,w_3w_4),(w_3,w_1w_2)\}$ that forms a lattice (see Figure \ref{reticulado1}).
It should be clear from the previous example that an important global property of operator $T$ is that starting from any worker-quasi-stable matching we always reach stability in a finite numbers of iterations. 

Now we are in a position to give an alternative proof of the existence of stable matchings and their lattice structure. To this end, we apply Tarski's fixed point theorem to the lattice of  worker-quasi-stable matchings. Remember that Tarski's theorem \citep{tarski1955lattice} states that if $(\mathcal{L},\geq )$ is a complete lattice and $T:\mathcal{L}\longrightarrow\mathcal{L}$ is isotone,  then the set of fixed points of $T$ is non-empty and forms a complete lattice with respect to $\geq$.

\begin{theorem}
The set of stable matchings is non-empty and forms a lattice with respect to $\succeq ^B$.
\end{theorem}
\begin{proof}
Let us check that operator $T$ verifies the hypothesis of Tarski's theorem.  First, notice that the lattice of worker-quasi-stable matchings is finite and therefore complete. Second, by Theorem \ref{Proposition punto fijo} (i), $T$  maps the lattice of worker-quasi-stable matchings to itself. Finally, $T$ is isotone by Lemma \ref{T isotone}. Then, by Tarski's theorem, the set of fixed points of $T$ is non-empty and forms a lattice under $\succeq ^B$. Moreover, by Theorem \ref{Proposition punto fijo} (iii), the set of fixed points of our operator $T$ is the set of stable matchings. 
\end{proof}

Another non-constructive argument for proving the existence of stable matchings can be provided following the lines of  \cite{sotomayor1996non}. By using the notion of  simple matchings in a one-to-one model, that paper shows the existence of stable matchings. Since a worker-quasi-stable matching is a generalization of  a simple matching, our operator $T$ provides a similar proof in a many-to-one model. To see this, notice that as the set $\mathcal{Q}$ is non-empty and finite, there is a maximal element $\mu$ for the partial order $\succeq ^B.$ We want to see that $\mu$ is a stable matching. If $\mu$ is not stable, then there is a blocking pair $(f,w)$  for $\mu$ in which $\mu(w)=\emptyset.$ Thus, $\mathcal{B}^\mu_\star\neq \emptyset$ and $T(\mu)\neq \mu.$ Since  $T$ is  a Pareto improving operator for the firms,  $T(\mu)\succ ^B \mu$, contradicting the maximality of $\mu$. Therefore, $\mu$  is a stable matching.

\subsection{A re-equilibration process via layoff chains}

In a labor market, sometimes new workers arrive or  firms downsize and laid off workers have to look for employment elsewhere. If we start from a stable matching and some of the previous situations are considered, we can model   this disruption to stability as a  worker-quasi-stable matching. 
In order to restore stability,  a re-equilibration process  can be described as a layoff  chain dynamic within worker-quasi-stable matchings that brings the market back to stability.\footnote{The notion of layoff chain is the counterpart for the workers' side to the notion of vacancy chain  for the firms' side studied by \cite{blum1997vacancy} and \cite{wu2018lattice}.} Each stage of this process can be modeled by applying our operator $T$  to a worker-quasi-stable matching $\mu$, as follows: 
\begin{enumerate}[(i)]
\item When new desirable workers become available in the market (unemployed workers in $\mu$), they propose to the most preferred firm among those  they can form a blocking pair with ($\mathcal{B}_{\star}^{\mu}$ in Subsection \ref{subsection operator T}).
\item Then, each firm $f$ selects the most preferred  subset of workers among those who just proposed to it ($\mathcal{W}^{\mu}_f$ in Subsection \ref{subsection operator T}) and its current employees.
\item Once the firms select their new sets of employees, a new set of unemployed workers becomes available for new proposals (unemployed workers in $T(\mu)$).
\end{enumerate}
Notice that, by Theorem \ref{Proposition punto fijo} (i) the sequence of matchings generated by this process belongs to the set of worker-quasi-stable matchings. Moreover, each matching in the sequence Pareto improves (for the firms) upon the previous matching in the sequence. Therefore, by the finiteness of the set of worker-quasi-stable matchings, this process reaches a fixed point of $T$, which is a stable matching by Theorem \ref{Proposition punto fijo} (iii).  

Thus, when a stable matching becomes a worker-quasi-stable matching due to market changes, the aforementioned process models how  a decentralized sequence of offers (in which unemployed workers are hired, causing new unemployments) produces a sequence of worker-quasi-stable matchings that converges to a stable matching.

\section{Further results with the Law of Aggregate Demand}\label{mas resultados LAD}

In this section, by requiring an additional condition on firms' choice functions, we can describe more accurately the re-equilibration process by means of the lattice structure of the set of  worker-quasi-stable matchings. This additional condition is the ``law of aggregate demand", that says that when a firm chooses from an expanded set, it hires at least as many workers as before. Formally, 
\begin{definition}
Choice function $C_f$  satisfies the\textbf{ law of aggregate demand (LAD)} if $S'\subseteq S\subseteq W$ implies $|C_f(S')|\leq |C_f(S)|.$
\end{definition}
We know that, starting from a worker-quasi-stable matching and iterating our operator $T$, we reach a fixed point of $T$. Assuming \textit{LAD}, the lattice structure can help us to identify this fixed point: it is the join of the original worker-quasi-stable matching and the worker-optimal stable matching $\mu_W$.\footnote{The set of stable matchings under substitutable choice functions is very well-structured. It contains two distinctive matchings: the firm-optimal stable matching $\mu_F$ and the worker-optimal stable matching $\mu_W$. The  matching $\mu_W$ is unanimously considered by all workers to be the best among all stable matchings and by
all firms to be the pessimal stable matching \cite[see][for more details]{roth1984evolution,blair1988lattice}.}
To formally present this result, for $\mu \in \mathcal{Q}$, let $\mathcal{F}(\mu)$ denote the fixed point of $T$  obtained by iterating it starting at $\mu$.
\begin{theorem}\label{thLAD}
 Let $\mu$ be a worker-quasi-stable matching. If firms' choice functions satisfy \textit{LAD}, then  $\mathcal{F}(\mu)=\mu \vee \mu_W.$
\end{theorem}

In order to prove Theorem  \ref{thLAD}, we first need to show that the join of a worker-quasi-stable matching and a stable matching is a stable matching (Proposition \ref{V es estable}).
 However, this is not true without \textit{LAD}. In Example   \ref{ejemplo 1}, where $\mu_W=\overline{\mu}$, we can observe that for $\mu'=(w_3,w_2)$ and $\overline{\mu}=(w_1w_2,w_3w_4),$ $\mu'\vee\overline{\mu}=(w_3,w_2w_4)$ which is not stable since $(f_2,w_1)$ is a blocking pair for $\mu'\vee\overline{\mu}.$ 
In this example firm $f_1$'s choice function, although substitutable, does not satisfy \textit{LAD}. To see this, observe that $C_{f_1}\left(\{w_1,w_2\}\right)=\{w_1,w_2\}$ whereas $C_{f_1}\left(\{w_1,w_2,w_3\}\right)=\{w_3\}.$ Nevertheless, if we restrict  firms' choice functions to satisfy substitutability and \textit{LAD}, we can recover this result.

\begin{proposition}\label{V es estable}
Let $\mu$ be a worker-quasi-stable matching, and $\mu'$ be a stable matching. If firms' choice functions  satisfy \textit{LAD}, then $\mu\vee\mu'$ is a stable matching.
\end{proposition}
\begin{proof}
Let $\mu$ be a worker-quasi-stable matching and $\mu'$ be a stable matching. Assume that $\mu\vee\mu'$ is not stable.  Thus there is a blocking pair $(f,w)$ such that $f\succ_w \mu \vee\mu'(w)$ and $w\in C_f\left(\mu\vee\mu'(f)\cup\{w\}\right).$ Since $\mu\vee\mu'$ is a worker-quasi-stable matching, $\mu\vee\mu'(w)=\emptyset$. By  (\ref{propiedad de choice}), $C_f\left(\mu\vee\mu'(f)\cup\{w\}\right)=C_f\left(\mu(f)\cup\mu'(f)\cup\{w\}\right)$. Then, by substitutability, $w\in C_f\left(\mu\vee\mu'(f)\cup\{w\}\right)$ implies  
\begin{equation}\label{ecu 2 lema V estable}
w\in C_f\left(\mu'(f)\cup\{w\}\right).
\end{equation}
If  $f\succ _{w} \mu'(w)$,  the pair $(f,w)$ blocks $\mu' ,$ contradicting the stability of $\mu'$. Therefore,  $\mu'(w)\succeq _{w} f$. Since $f\succ_w \mu \vee\mu'(w)= \emptyset$, the individual rationality of $\mu'$ implies that there is $\tilde{f}\in F$ (possibly $\tilde{f}=f$) such that $\tilde{f}= \mu'(w).$ Hence, $w\in \mu'(\tilde{f}).$ Then, 
\begin{equation}\label{ecu 3 lema V estable}
w\in \mu'(F),
\end{equation}
where $\mu'(F)= \bigcup _{g\in F}  \mu'(g)$. Also, since $\mu\vee\mu'(w)=\emptyset$, by the definition of  $\mu\vee\mu'$,
\begin{equation}\label{ecu 4 lema V estable}
w\notin \mu\vee\mu'(F).
\end{equation}
As $\mu'(g) \subseteq \mu(g) \cup \mu'(g)$ for each $g \in F,$  \emph{LAD} implies
$|C_g( \mu(g)\cup \mu'(g))|  \geq |C_g( \mu'(g))| .$
%\end{equation}
Therefore, by definition of $\mu \vee \mu'$ and the individual rationality of $\mu',$ 
$$ |\mu\vee\mu'(g)|  = |C_g( \mu(g)\cup \mu'(g))| \geq |C_g( \mu'(g))| = |\mu'(g) |$$
 for each $g \in F.$ Since $\mu\vee\mu'$ and $\mu'$ are matchings,
\begin{equation}\label{ecu 5 lema V estable}
\left| \mu\vee\mu'(F) \right|=\sum _{g\in F} | \mu\vee\mu'(g) | \geq \sum _{g\in F} | \mu'(g) |=\left| \mu'(F) \right|.
\end{equation}
By  (\ref{ecu 3 lema V estable}) and  (\ref{ecu 4 lema V estable}) $w\in \mu'(F)\setminus \mu\vee\mu'(F)$. This fact together with  (\ref{ecu 5 lema V estable}) imply that there is $w'\in W$ such that  $w'\in \mu\vee\mu'(F)\setminus \mu'(F)$. Then, $\mu'(w')=\emptyset $ and there is $f'\in F$ such that $w'\in C_{f'}\left(\mu(f')\cup \mu'(f')\right).$ By substitutability,  $w'\in C_{f'}\left(\mu'(f')\cup \{w'\} \right).$ Furthermore, individual rationality of $\mu'$ implies that $f'\succ _{w'} \emptyset$. Hence, the pair $(f',w')$ blocks $\mu'$. This contradicts the stability of $\mu'$. Therefore,  $\mu\vee\mu'$ is a stable matching.  
\end{proof}

The next corollary, that follows easily  from Proposition \ref{V es estable}, presents an important feature of the structure of the worker-quasi-stable matching set. It states that any worker-quasi-stable matching that Blair-dominates $\mu_W$ is actually a stable matching. This happens because the join between them, that is equal to the worker-quasi-stable matching, is a stable matching by Proposition \ref{V es estable}. 
\begin{corollary}
Let $\mu$ be a worker-quasi-stable matching. If  firms' choice functions satisfy \textit{LAD} and $\mu \succeq ^B\mu_W$, then $\mu$ is a stable matching.
\end{corollary}
 
 This corollary makes a deep connection between worker-quasi-stable matchings and stable matchings. Notice that, when choice functions satisfy substitutability and \textit{LAD}, the set of stable matchings has a dual lattice structure \citep[see for instance][for more details]{alkan2002class}. Therefore, by duality, $\mu_W$ is the firm-pessimal stable matching. The implications of this are twofold:  (i) a worker-quasi-stable matching that Blair-dominates \textit{any} stable matching is also stable, and (ii) a matching that is worker-quasi-stable but not stable is  either Blair-incomparable to or  Blair-dominated by $\mu_W.$

Now we are in a position to prove the main result of this section.\smallskip

\noindent \begin{proof}[Proof of Theorem \ref{thLAD}]
Let $\mu \in \mathcal{Q}.$ By Proposition \ref{V es estable}, $\mu \vee \mu_W$ is a stable matching. By Lemma \ref{T isotone} and Theorem \ref{Proposition punto fijo} (iii),  $\mu \vee \mu_W \succeq ^B \mu$ implies that $\mu \vee \mu_W \succeq ^B \mathcal{F}(\mu)$. Moreover, given that $T$ is a weakly Pareto improving operator by Theorem \ref{Proposition punto fijo} (ii), $\mathcal{F}(\mu) \succeq ^B \mu$. Since $\mu_W$ is the firms' pessimal stable matching, $\mathcal{F}(\mu) \succeq ^B \mu_W.$ As $\mathcal{F}(\mu) \succeq ^B \mu$ and  y $\mathcal{F}(\mu) \succeq ^B \mu_W,$ by definition of join,  $\mathcal{F}(\mu) \succeq ^B \mu \vee\mu_W$. Thus, by antisymmetry, $\mathcal{F}(\mu) = \mu \vee\mu_W.$
\end{proof}

The following proposition provides an upper bound for the total number of workers hired in any worker-quasi-stable matching by a firm. 
If firms' choice functions satisfy \textit{LAD}, this number will not exceed the total number of workers matched by that firm in any stable matching.
\begin{proposition}
Let $\mu$ be a worker-quasi-stable matching and let $\mu'$ be a stable matching. If firms' choice functions satisfy \textit{LAD}, $|\mu(f)|\leq |\mu'(f)|$ for each $f\in F.$
\end{proposition}
\begin{proof}
Let $\mu \in \mathcal{Q}$ and $f \in F.$ By \emph{LAD} and individual rationality of $\mu,$ we have that $|T(\mu)(f)|=|C_f(\mu(f) \cup \mathcal{W}^\mu_f)|\geq|C_f(\mu(f))| = |\mu(f)|.$ Iterating we obtain that $|\mathcal{F}(\mu)(f)| \geq |\mu(f)|.$ By definition,  $\mathcal{F}(\mu)$ is a fixed point of our Tarski operator. Thus, $\mathcal{F}(\mu)$ is stable by Theorem \ref{Proposition punto fijo} (iii). Then, by the Rural Hospital Theorem\footnote{The \textit{Rural Hospital Theorem} is proven in different contexts by many authors \citep[see][among others]{mcvitie1970stable,roth1984evolution,roth1985college,martinez2000single,alkan2002class,kojima2012rural}.
The version of this theorem for a many-to-many matching market where all agents have substitutable choice functions satisfying \emph{LAD}, that also applies in our setting, is presented in \cite{alkan2002class} and states that each agent is matched with the same number of partners in every stable matching.} $|\mathcal{F}(\mu)(f)|=|\mu'(f)|$ for each $\mu' \in \mathcal{S}$ and each $f \in F.$ Therefore,  $|\mu(f)|\leq|\mathcal{F}(\mu)(f)|=|\mu'(f)|$ for each $\mu' \in \mathcal{S}$ and each $f \in F.$ 
\end{proof}

Finally, even though the results that we obtain in this section under substitutability and \textit{LAD} are analogous to those obtained by \cite{wu2018lattice} for firm-quasi-stable matchings under responsive preferences (their Lemma 3.11, Theorem 3.12, and Corollaries 4.1 and 4.2), there are important differences. One of them is that our setting is more general, since each substitutable choice function induces a substitutable preference, and the class of substitutable preferences  contains the class of responsive preferences. Another one is that their method of proof relies on the results established by \cite{blum1997vacancy}, that under responsiveness link many-to-one models with one-to-one models, while ours is  independent and based on Blair's approach under substitutability.

\section{Conclusions}\label{concludings}

This paper presents the set of worker-quasi-stable matchings as an  extension of the set of stable matchings for a many-to-one matching model. We show that the set of worker-quasi-stable matchings is well-structured: it forms a full lattice with respect to Blair's partial order. Relying on the lattice structure of this set,  a  re-equilibration process that models a layoff chain and that converges to stability is also described. Notice that even though we explicitly construct the join between any pair of worker-quasi-stable matchings,  how to compute the meet is still  an open question.  
 
The results in this paper can be generalized straightforwardly to a model of matching with contracts as in \cite{hatfield2005matching}. This extension does not present any difficulties other than the notational ones proper to this setting.

 It is usual in the literature to study many-to-one models assuming that firms' preferences are responsive. This is due to the close relation between this model with responsive preferences and the one-to-one model \citep[for a thorough survey on this fact, see][]{roth1992two}.  However, when firms are endowed with substitutable choice functions (a much less restrictive requirement), this relation with the one-to-one model no longer holds. Then, to extend results to the many-to-one model under the substitutable assumption,  Blair's partial order becomes crucial.  We believe that to exploit  Blair's approach  for extending other known results in models with responsive preferences to models with substitutable choice functions is a natural direction to pursue further research.

\bibliographystyle{ecta}

\end{document}